\documentclass[envcountsame]{llncs}
\usepackage{graphicx}
\usepackage{amssymb}
\usepackage{amsmath}
\usepackage{enumerate}
\usepackage{todonotes}
\usepackage[normalem]{ulem}
\usepackage{subfig}
\usepackage{color}
\usepackage{ifsym}

\newcommand{\singlevertical}{single-vertical}

\newcommand{\component}{object}
\newcommand{\components}{objects}
\newcommand{\m}{{\bf m}}

\newtheorem{observation}{Observation}

\title{Splitting $B_2$-VPG graphs into
	outer-string and co-comparability graphs}
\author{Therese Biedl\thanks{The author was supported by NSERC.} \and Martin Derka\thanks{The author was supported by Vanier CGS.}}
\institute{David R.~Cheriton School of Computer Science,\\ University of Waterloo, Waterloo, ON, N2L 3G1, Canada\\
\email{\{biedl,mderka\}@uwaterloo.ca}
}

\begin{document}
\maketitle

\begin{abstract}
In this paper, we show that any $B_2$-VPG graph (i.e., an intersection
graph of orthogonal curves with at most 2 bends) can be decomposed into
$O(\log n)$ outerstring graphs or $O(\log^3 n)$ permutation graphs.  This
leads to better approximation algorithms for hereditary graph problems, such
as independent set, clique and clique cover, on $B_2$-VPG graphs.
\end{abstract}

\section{Preliminaries}

An \emph{intersection representation} of a graph is a way of portraying a graph
using geometric objects. In such a representation, every object 
corresponds to a vertex in the graph, and there is an edge between vertices $u$ and $v$
if and only if their two objects $\mathbf{u}$ and $\mathbf{v}$ intersect. 

One of the well-studied classes of such intersection graphs 
are the \emph{string graph}, where the objects are (open) curves in the plane.
An \emph{outer-string} representation is one
where all the curves are in inside a polygon $P$ and touch the
boundary of $P$ at least once.    
A string representation is called a {\em 1-string representation} if any
two strings intersect at most once.
It is called a {\em $B_k$-VPG-representation}
(for some $k\geq 0$) if every curve is an orthogonal curve with 
at most $k$ bends.    We naturally use the term {\em outer-string graph}
for  graphs that have an outer-string representation, and similarly for
other types of intersecting objects.

\paragraph{Our contribution:} 
This paper is concerned with partitioning string graphs (and other classes
of intersection graphs) into subgraphs that have nice properties, such as
being outerstring graphs or permutation graphs (defined formally below).  
We can then use such a partition
to obtain approximation algorithms for some graph problems, such as weighted
independent set, clique, clique cover and colouring.  More specifically,
``partitioning'' in this paper usually means a {\em vertex partition}, i.e.,
we split the vertices of the graph as $V=V_1\cup \dots \cup V_k$ such
that the subgraph induced by each $V_i$ has nice properties.  In one
case we also do an {\em edge-partition} where we partition $E=E_1\cup E_2$
and then work on the two subgraphs $G_i=(V,E_i)$.

Our paper was inspired by a paper by
Lahiri et al.~\cite{cit:cocoa} in 2014, which
gave an algorithm to approximate the maximum (unweighted) independent
set in a $B_1$-VPG graph within a factor of $4\log^2 n$.  We greatly
expand on their approach as follows.  First,
rather than solving maximum independent set directly, we instead
split such a graph into subgraphs. This allows
us to approximate not just independent set, but more generally any hereditary
graph problem that is solvable in such graphs.

Secondly, rather than using co-comparability graphs for splitting as Lahiri
et al.~did, we use outerstring graphs.  This allows us to stop the splitting
earlier, reducing the approximation factor from $4\log^2 n$ to $2\log n$,
and to give an algorithm for {\em weighted} independent set (wIS).

Finally, we
allow much more general shapes.  For splitting into outerstring graphs,
we can allow any shape that can
be described as the union of one vertical and any number of horizontal
segments (we call such intersection graphs {\singlevertical}).  
Our results imply a $2\log n$-approximation algorithm for wIS
in such graphs, which include $B_1$-VPG graphs,
and a $4\log n$-approximation for wIS in $B_2$-VPG graphs.

\smallskip

In the second part of the paper, we consider splitting the graph such
that the resulting subgraphs are co-comparability graphs.    This type
of problem was first considered by Keil and Stewart \cite{KeilStewart2006},
who showed that so-called subtree filament graphs can be vertex-partitioned
into $O(\log n)$ co-comparability graphs.  The work of Lahiri et 
al.~\cite{cit:cocoa} can be seen as proving that every $B_1$-VPG graph can
be vertex-partitioned into $O(\log^2 n)$ co-comparability graphs.
We focus here on the super-class of $B_2$-VPG-graphs, and show that
they can be vertex-partitioned into $O(\log^3 n)$ co-comparability graphs.
Moreover, these co-comparability graphs have poset dimension 3, and if
the $B_2$-VPG representation was 1-string, then they are 
permutation graphs.  This leads to better approximation algorithms
for clique, colouring and clique cover for $B_2$-VPG graphs.

\section{Decomposing into outerstring graphs}

We argue in this section how to split a graph into outerstring graphs
if it has an intersection representation of a special form.
A {\em single-vertical object} is a connected set 
$S\subset \Bbb{R}^2$ of the form
$S=s_0\cup s_1\cup \dots \cup s_k$, where $s_0$ is a vertical segment and
$s_1,\dots,s_k$ are horizontal segments, for some finite $k$.
Given a number of single-vertical objects $S_1, \dots, S_n$, we define
the intersection graph of it in the usual way, by defining one vertex
per object and adding an edge whenever objects have at least one point in
common (contacts are considered intersections).    
We call such a representation a \emph{{\singlevertical} representation} and the 
graph a \emph{{\singlevertical} intersection graph}. 
The \emph{$x$-coordinate} of one {\singlevertical} {\component}
is defined to be the $x$-coordinate of the (unique) vertical segment. 
We consider a horizontal segment to be a {\singlevertical} object as well,
by attaching a zero-length vertical segment at one of its endpoints.

\begin{theorem}
\label{thm:main}
Let $G$ be a {\singlevertical} intersection graph.
Then the vertices of $G$
can be partitioned into at most $\max\{1,2\log n\}$%
\footnote{This bound is not tight; a more careful analysis shows
that we get at most $\max\{1,2\lceil \log n\rceil -2\}$ graphs.}
sets such that the subgraph
induced by each is an outer-string graph.
\end{theorem}

Our proof of Theorem~\ref{thm:main} uses a splitting technique implicit in the
the recursive approximation algorithm of Lahiri et al.~\cite{cit:cocoa}.
Let $R$ be a {\singlevertical} representation on $G$ and $S$ be an ordered 
list of the $x$-coordinates of all the {\components} in $R$. 
We define the \emph{median} $m$ of $R$ as the
smallest number such that at most $\frac{|S|}{2}$ $x$-coordinates in $S$ are
smaller than $m$ and
at most $\frac{|S|}{2}$ $x$-coordinates in $S$ are
bigger than $m$.   (If $|S|$ is odd then $m$ is hence the $x$-coordinate
of at least one object.)
Now split $R$ into three sets: The \emph{middle} set $M$ of {\components}
that intersect the vertical line {\m} with $x$-coordinate $m$; 
the \emph{left} set $L$ of {\components}
whose $x$-coordinates are smaller than $m$ and that do not belong to $M$, and the \emph{right} 
set $R$ of {\components} whose $x$-coordinates are bigger than $m$ and that do not belong to $M$.
Split $M$ further into $M_L$ = \{ $c$ {\textbar} the $x$-coordinate of $c$ is
less than $m$\} and $M_R = M \setminus M_L$. 

\begin{figure}
\centering
\includegraphics[width=.8\textwidth]{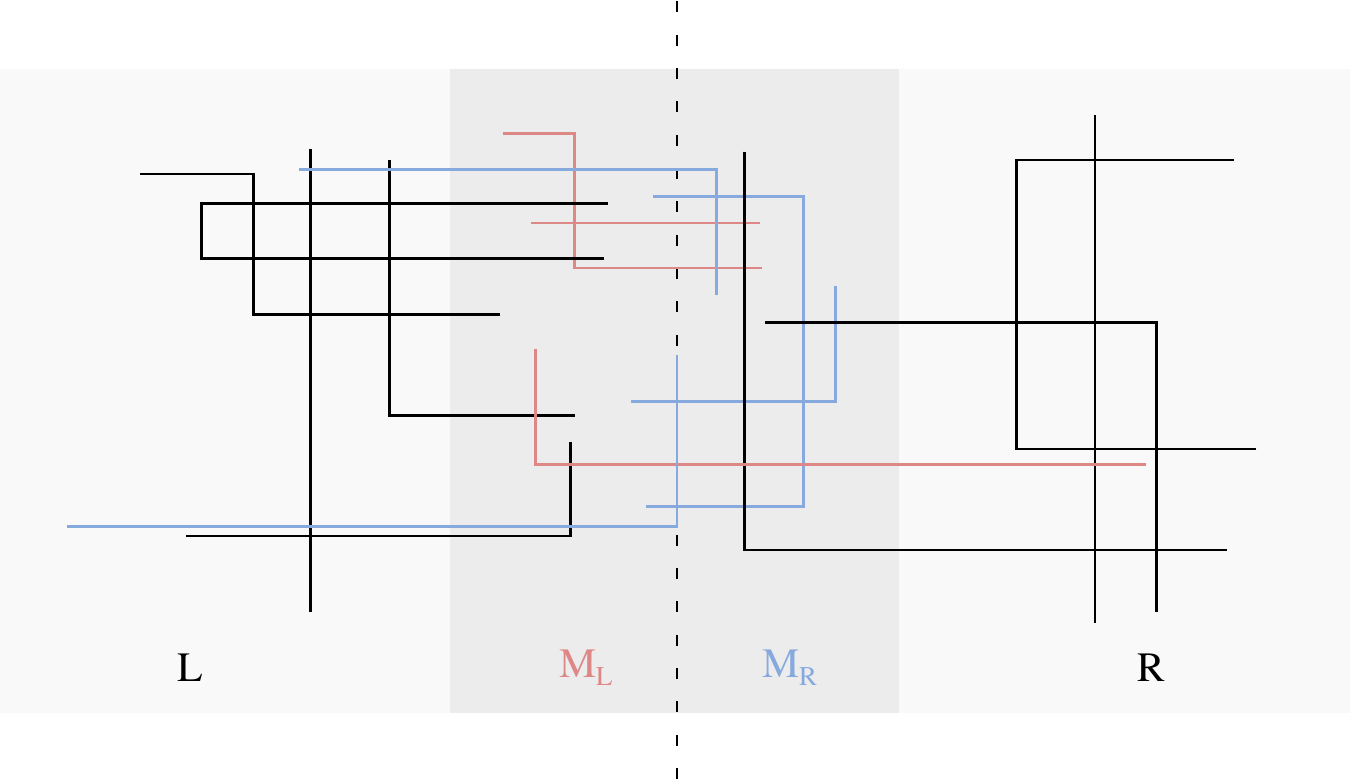}
\caption{The split of a representation into $L$, $M = M_L \cup M_R$ and $R$.}
\end{figure}

\begin{lemma}
\label{claim:outer-string}
The subgraph induced by the {\components} in $M_L$ is outer-string.
\end{lemma}
\begin{proof}
All the {\components} in $M_L$ intersect curve \m. Since all the $x$-coordinates of
those {\components} are smaller than $m$, all the intersections of the {\components} occur
left of \m. If an {\component} is not a curve, one can replace it by a closed curve that traces the
shape of the {\component} left of {\m}. Breaking the closed trace-curve at one of the attachments to {\m}
produces an open curve. Doing so for every {\component} that is not a curve, one obtains an outer-string 
representation where all curves attach to {\m} from one side and that induces the same graph 
as $M_L$. 
\qed
\end{proof}

A similar proof shows that the graph induced by {\components} in $M_R$ 
is an outerstring graph.
Now we can prove our main result:

\begin{proof}[of Theorem~\ref{thm:main}]
Let $G$ be a graph with a {\singlevertical} representation. We proceed by 
induction on the number of vertices $n$ in $G$. 
If $n \leq 2$, then 
the graph is outer-string and we are done, so assume $n
\geq 3$, which implies that $\log n\geq \frac{3}{2}$.
By Lemma~\ref{claim:outer-string}, both $M_L$ and $M_R$ individually induce
an outer-string graph.
Applying induction, we get at most 
$$\max\{1,2\log|L|\} \leq \max\{1,2\log(n/2)\} = \max\{1,2\log n-2\} 
= 2\log n-2$$
outer-string
subgraphs for $L$, and similarly at most $2\log n-2$ outerstring
subgraphs for $R$.  Since
the {\components} in $L$ and $R$ are separated by the vertical 
line \m, there are no edges between the corresponding vertices.
Thus any outerstring subgraph defined by $L$ can be combined with any
outerstring subgraph defined by $R$ to give one outerstring graph.
We hence obtain $2\log n-2$
outerstring graphs from recursing into $L$ and $R$.  Adding to this the
two outer-string graphs defined by $M_L$ and $M_R$ gives the result.
\qed
\end{proof}

Our proof is constructive, and finds the partition within $O(\log n)$
recursions.  In each recursion we must find the median $m$ and then
determine which objects intersect the line \m.  If we pre-sort three
lists of the objects (once by $x$-coordinate of the vertical segment,
once by leftmost $x$-coordinate, and once by rightmost $x$-coordinate),
and pass these lists along as parameters,
then each recursion can be done in $O(n)$ time, without linear-time
median-finding.  The pre-sorting takes $O(N+n\log n)$ time, where $N$
is the total number of segments in the representation.  Hence the run-time
to find the partition is $O(N+n\log n)$.

The above results were for {\singlevertical} graphs.  However, the
main focus of this paper is $B_k$-VPG-graphs, for $k\leq 2$.  Clearly
$B_1$-VPG graphs are {\singlevertical} by definition.  But $B_2$-VPG-graphs
are not obviously {\singlevertical}, since they might use curves in form
of a $U$, with two vertical segments.
However, we can still handle them by doubling the number of graphs into
which we split.

\begin{lemma}
\label{lem:B2VPGsinglevertical}
Let $G$ be a $B_2$-VPG graph.
Then the vertices of $G$
can be partitioned into $2$ sets such that the subgraph
induced by each is a {\singlevertical} $B_2$-VPG graph.
\end{lemma}
\begin{proof}
Fix a $B_2$-VPG-representation of $G$.
Let $V_v$ be the vertices that have at most one vertical
segment in their curve, and $V_h$ be the remaining vertices.  Since
every curve has at most three segments, and all curves in $V_h$ have at
least two vertical segments, each of them has at most one horizontal segment.
Clearly $V_v$ induces a {\singlevertical} graph.
$V_h$ {\em also}
induces a {\singlevertical} graph, because we can rotate all curves by 
$90^\circ$ and then have at most one vertical segment per curve.  
\qed
\end{proof}

Combining this with Theorem~\ref{thm:main}, we immediately obtain:

\begin{corollary}
\label{thm:B2VPGouterstring}
Let $G$ be a $B_2$-VPG graph.  Then the vertices of $G$
can be partitioned into at most $\max\{1,4\log n\}$%
sets such that the subgraph
induced by each is an outerstring graph.
\end{corollary}

\section{Decomposing into co-comparability graphs}

We now show that by doing further splits, we can actually decompose
$B_2$-VPG graphs into so-called co-comparability graphs of poset
dimension 3 (defined formally below).  While we require more subgraphs
for such a split, the advantage is that numerous problems are polynomial
for such co-comparability graphs, while for outerstring we know of no
problem other than weighted independent set that is poly-time solvable.

We first give an outline of the approach.  Given a $B_2$-VPG-graph,
we first use Lemma~\ref{lem:B2VPGsinglevertical} to split it into
two {\singlevertical} $B_2$-VPG-graphs.  Given a {\singlevertical}
$B_2$-VPG-graph, we next use a technique much like the one of
Theorem~\ref{thm:main} to split it into $\log n$ {\singlevertical}
$B_2$-VPG-graphs that are ``centered'' in some sense.  Any such
graph can easily be edge-partitioned into two $B_1$-VPG-graphs
that are ``grounded'' in some sense.  We then apply the technique
of Theorem~\ref{thm:main} again (but in the other direction) to
split a grounded $B_1$-VPG-graph into $\log n$ $B_1$-VPG-graphs
that are ``cornered'' in some sense.  The latter graphs can
be shown to be permutation graphs.   This gives the result after
arguing that the edge-partition can be un-done at the cost of
combining permutation graphs into co-comparability graphs.

\subsection{Co-comparability graphs}

We start by defining the graph classes that we use in this section only.
A graph $G$ with vertices $\{1,\dots,n\}$ is called a {\em permutation
graph} if there exists two permutations $\pi_1,\pi_2$ of $\{1,\dots,n\}$ such
that $(i,j)$ is an edge of $G$ if and only if $\pi_1$ lists $i,j$ in the
opposite order as $\pi_2$ does.  Put differently, if we place $\pi_1(1),\dots,\pi_1(n)$
at points along a horizontal line, and $\pi_2(1),\dots,\pi_2(n)$ at points along a
parallel horizontal line, and use the line segment $(\pi_1(i),\pi_2(i))$
to represent vertex $i$, then the graph is the intersection graph of these segments.

A {\em co-comparability graph} $G$ is a graph whose complement can be directed
in an acyclic transitive fashion.  Rather than defining these terms, we
describe here only the restricted type of co-comparability graphs that we are 
interested in.  A graph $G$ with vertices $\{1,\dots,n\}$
is called a {\em co-comparability graph of poset dimension $k$}
if there exist $k$ permutations $\pi_1,\dots,\pi_k$ such that $(i,j)$ is an edge
if and only if there are two permutations that list $i$ and $j$ in opposite order.  (See Golumbic et al.~\cite{GolumbicRU1983} for more on these
characterizations.)
Note that a permutation graph is a co-comparability graph of poset dimension 2.

\begin{figure}
\includegraphics[width=\textwidth]{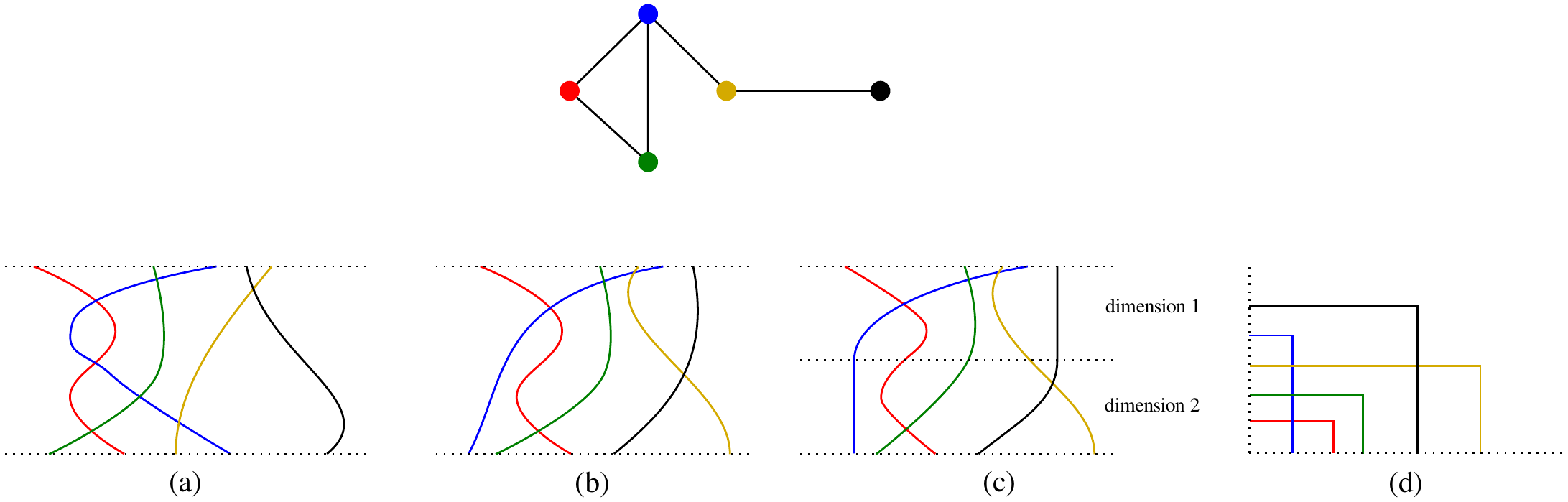}
\caption{A graph that is simultaneously (a) a co-comparability graph; (b) a permutation graph;
(c) a co-comparability graph of poset dimension 2; and (d) a cornered $B_1$-VPG graph.}
\label{fig:kinds}
\end{figure}

\subsection{Cornered $B_1$-VPG graphs}

A $B_1$-VPG-representation is called {\em cornered} if there exists
a horizontal and a vertical ray emanating from the same point such that
any curve of the representation intersects both rays.  See 
Fig.~\ref{fig:kinds}(d) for an example.

\begin{lemma}
\label{lem:permutation}
If $G$ has a cornered $B_1$-VPG-representation, say with
respect to rays $r_1$ and $r_2$, then
$G$ is a permutation graph.    Further, the two permutations
defining $G$ are exactly the two orders in which vertex-curves 
intersect $r_1$ and $r_2$.
\end{lemma}
\begin{proof}
Since the curves have only one bend, the intersections with $r_1$ and $r_2$
determine the curve of each vertex.
In particular, two curves
intersect if and only if the two orders along $r_1$ and $r_2$ is {\em not}
the same, which is to say, if their orders are different in the two
permutations of the vertices defined by the orders along the rays.  Hence
using these orders show that $G$ is a permutation graph.
\qed
\end{proof}

\subsection{From grounded to cornered}

We call a $B_1$-VPG representation \emph{grounded} if there exists a 
horizontal line segment $\ell_H$ that intersects the all curves, and 
has all horizontal segments of all curves above it.
See also Fig.~\ref{fig:combine} and \cite{CFM+16} for more
properties of graphs that have a grounded representation.
We now show how to split a grounded $B_1$-VPG-representation into 
cornered ones.    It will be important later that not only can we
do such a split, but we know how the curves intersect $\ell_H$ afterwards.
More precisely, the curves in the resulting representations may not be
identical to the ones we started with, but they are modified only in
such a way that the intersections points of curves along $\ell_H$ is
unchanged.

\begin{lemma}
\label{lem:main_ground}
Let $R$ be a $B_1$-VPG-representation that is grounded with 
respect to segment $\ell_H$.
Then $R$ can be partitioned into at most $\max\{1,2\log n\}$ sets 
$R_1,\dots,R_K$ such that each set $R_i$ is cornered after upward translation
and segment-extension of some of its curves.
\end{lemma} 
\begin{proof}
A single curve with one bend is always cornered, so the claim is easily
shown for $n\leq 4$ where $\max\{1,2\log n\}\geq n$.  For $n\geq 5$,
it will be helpful to split $R$ first
into two sets, those curves of the form {\scriptsize\textifsym{|h}} and those that form
{\scriptsize\textifsym{h|}} (no other shapes can exist in a grounded $B_1$-VPG-representation).
The result follows if we show that each of them can be split into
$\log n$ many cornered $B_1$-VPG-representations.

So assume that $R$ consists of only {\scriptsize\textifsym{|h}}'s.
We apply essentially the same idea
as in Theorem~\ref{thm:main}.  
Let again {\m} be the
vertical line along the median of $x$-coordinates of vertical
segments of curves.  Let $M$ be all those curves that intersect {\m}.
Since curves are {\scriptsize\textifsym{|h}}'s, 
any curve in $M$ intersects $\ell_H$ to the left of {\m},
and intersects {\m} above $\ell_H$.  
Hence taking the two rays along $\ell_H$ and {\m}
emanating from their common point shows that $M$ is cornered.

\begin{figure}
\includegraphics[width=\textwidth]{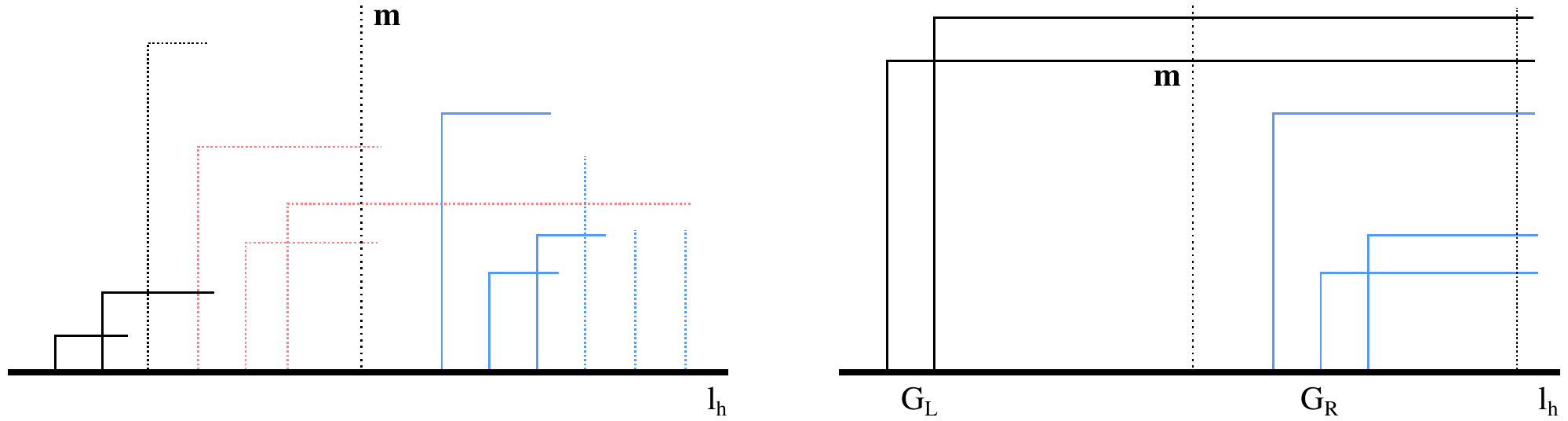}
\caption{An illustration for the proof of Lemma~\ref{lem:main_ground}.
(left) Splitting a cornered $B_1$-VPG graph. (right) Combining a graph $G_L$ with a graphs $G_R$
so that the result is a cornered $B_1$-VPG graph.}
\label{fig:combine}
\end{figure}


We then recurse both in the subgraph $L$ of vertices entirely left of {\m}
and the subgraph $R$ of vertices entirely right of {\m}.  Each of them is split
recursively into at most $\max\{1,\log (n/2)\}=\log n - 1$ 
subgraphs that are cornered.
We must now argue how to combine two such subgraphs $G_L$ and $G_R$ (of
vertices from $L$ and $R$) such that they are cornered while modifying
curves only in the permitted way.

Translate curves of $G_L$ upward such that the lowest
horizontal segment of $G_L$ is above the highest horizontal segment
of $G_R$.  Extend the vertical segments of $G_L$ so that they again
intersect $\ell_H$.  Extend horizontal segments of both $G_L$ and $G_R$ 
rightward until they all intersect one vertical line segment. 
The resulting representation satisfies all conditions.  

Since we obtain at most $\log n-1$ such cornered representations
from the curves in $R\cup L$, we can add $M$ to it and the result follows.
\qed
\end{proof}

\begin{corollary}
\label{cor:B1cornered}
Let $G$ be a graph with a grounded $B_1$-VPG representation.  
Then the vertices of $G$
can be partitioned into at most $\max\{1,2\log n\}$ 
sets such that the subgraph
induced by each is a permutation graph.
\end{corollary}

\subsection{From centered to grounded}

We now switch to VPG-representations with 2 bends, but currently only
allow those with a single vertical segment per curve.  So let $R$ be a {\singlevertical}
$B_2$-VPG-representation.  We call $R$ {\em centered} if there exists a
horizontal line segment $\ell_H$
that intersects the vertical segment of all curves.
Given such a representation, we can cut each curve apart at the intersection
point with $\ell_H$.  Then the parts above $\ell_H$ form a grounded
$B_1$-VPG-representation, and the parts below form (after a 180$^\circ$
rotation) also a grounded $B_1$-VPG-representation.   Note that this
split corresponds to splitting the edges into $E=E_1\cup E_2$, depending
on whether the intersection for each edge occurs above or below $\ell_H$.
Note that if curves may intersect repeatedly, then an edge may be in both sets. See 
Fig.~\ref{fig:center} for an example.
With this, we can now split into co-comparability graphs.

\begin{figure}
\centering
\includegraphics[width=.8\textwidth]{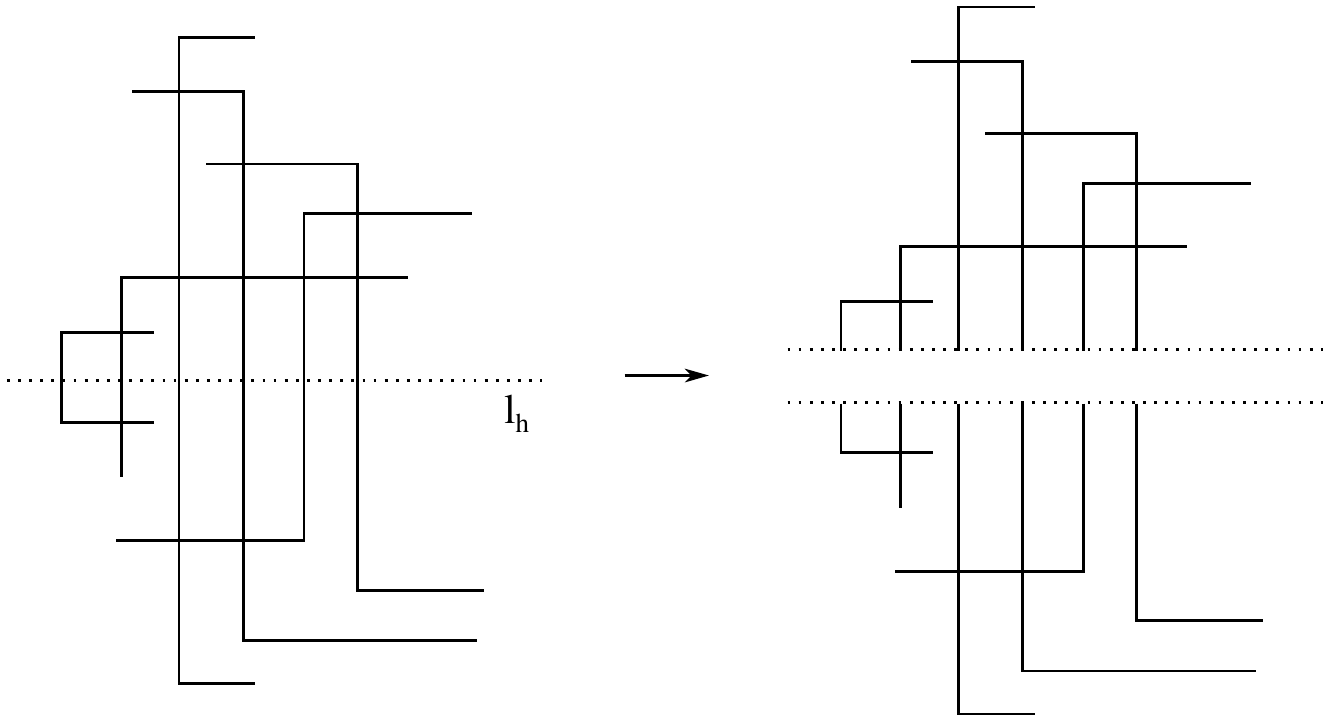}
\caption{Splitting singlevertical $B_2$-VPG-representation into two grounded $B_1$-VPG-representations.}
\label{fig:center}
\end{figure}

\begin{lemma}
\label{lem:cocompar}
Let $G$ be a graph with a {\singlevertical} centered $B_2$-VPG representation.  
Then the vertices of $G$
can be partitioned into at most $\max\{1,4\log^2 n\}$ 
sets such that the subgraph
induced by each is a co-comparability graph of poset dimension 3.
\end{lemma}
\begin{proof}
The claim clearly holds for $n\leq 4$, so assume $n\geq 5$.
Let $\ell_H$ be the horizontal segment along which the representation is centered.
Split the edges into $E_1$ and $E_2$ as above, and let $R_1$ and $R_2$ be the resulting
grounded $B_1$-VPG-representations, which have the same order of vertical intersections
along $\ell_H$.  Split $R_1$ into $K\leq 2\log n$ sets of curves $R_1^1,\dots,R_1^K$, each of which forms a
cornered $B_1$-VPG-representation that uses the same order of intersections
along $\ell_H$.  Similarly split $R_2$ into $K'\leq 2\log n$ sets $R_2^1,
\dots,R_2^{K'}$ of 
cornered $B_1$-VPG-representations.

Now define $R_{i,j}$ to consist of all those curves $r$ where the part of
$r$ above $\ell_H$ belongs to $R_1^i$ and the part below belongs to $R_2^j$.
This gives $K\cdot K'\leq 4\log^2 n$ sets of curves.
Consider one such set $R_{i,j}$.  The parts of curves in $R_{i,j}$ that
were above $\ell_H$ are cornered at $\ell_H$ and some vertical upward ray,
hence define a permutation $\pi_1$ along the vertical ray and $\pi_2$
along $\ell_H$. 
Similarly
the parts of curves below $\ell_H$ define two permutations, say $\pi_2'$ 
along $\ell_H$ and $\pi_3$ along some vertical downward ray.
But the split into cornered $B_1$-VPG-representation ensured that the
intersections along $\ell_H$ was not changed, so $\pi_2=\pi_2'$.
The three permutations $\pi_1,\pi_2,\pi_3$ together hence define a 
co-comparability graph of poset dimension 3 as desired.
\qed
\end{proof}

We can do slightly better if the representation is additionally 1-string.

\begin{corollary}
\label{cor:permutation}
Let $G$ be a graph with a {\singlevertical} centered 1-string 
$B_2$-VPG representation.  Then the vertices of $G$
can be partitioned into at most $\max\{1,4\log^2 n\}$ 
sets such that the subgraph
induced by each is a permutation graph.
\end{corollary}
\begin{proof}
The split is exactly the same as in Lemma~\ref{lem:cocompar}.
Consider one of the subgraphs $G_i$ and the
permutations $\pi_1,\pi_2,\pi_3$ that came with it, where $\pi_2$ is
the permutation of curves along the centering line $\ell_H$.
We claim that $G_i$ is a permutation graph, using $\pi_1,\pi_3$ as the 
two permutations.   Clearly if $(u,v)$ is not an edge of $G_i$,
then all of $\pi_1,\pi_2,\pi_3$ list $u$ and $v$ in the same order.
If $(u,v)$ is an edge of $G_i$, then two of $\pi_1,\pi_2,\pi_3$ list
$u,v$ in opposite order.  We claim that $\pi_1$ and $\pi_3$ list $u,v$
in opposite order.  For if not, say $u$ comes before $v$ in both $\pi_1$
and $\pi_3$, then (to represent edge $(u,v)$) we must have $u$ after
$v$ in $\pi_2$.  But then the curves of $u$ and $v$ intersect both above
and below $\ell_H$, contradicting that we have a 1-string representation.
So the two permutations $\pi_1,\pi_3$ define graph $G_i$.
\qed
\end{proof}

\subsection{Making {\singlevertical} $B_2$-VPG-representations centered}

\begin{lemma}
\label{lem:centered}
Let $G$ be a graph with a {\singlevertical} $B_2$-VPG representation.  Then the vertices of $G$
can be partitioned into at most $\max\{1,\log n\}$ 
sets such that the subgraph
induced by each has a {\singlevertical} centered $B_2$-VPG-representation.
\end{lemma}
\begin{proof}
The approach is quite similar to the one in Theorem~\ref{thm:main}, but uses a
horizontal split and a different median.  The claim is easy to show for $n=3$,
so assume $n\geq 4$. Recall that there are are $n$
vertical segments, hence $2n$ endpoints of such segments.
Let $m$ be the value such that at most $n$ of these endpoints each are
below and above $m$, and let {\m} be the horizontal
line with $y$-coordinate $m$.

Let $M$ be the curves that are intersected by {\m}; clearly they form a 
{\singlevertical} centered $B_2$-VPG-representation.  Let $B$ be all those
curves whose vertical segment (and hence the entire curve) is completely
below {\m}.  Each such curve contributes two endpoints of vertical segments,
hence $|B|\leq n/2$ by choice of $m$.  Recursively
split $B$ into at most $\max\{1,\log (n/2)\}=\log n-1$ sets, and 
likewise split the curves $U$
above {\m} into at most $\log n-1$ sets.

Each chosen subset $G_B$ of $B$ is centered, 
as is each chosen subset $G_U$ of $U$.
Since $G_B$ uses curves below {\m} while $G_U$ uses curves above, 
there are no crossings between
these curves.  We can hence translate the curves of $G_B$ such they are 
centered with the same horizontal line as $G_U$.  
Therefore $G_B \cup G_U$ has a centered
{\singlevertical} $B_2$-VPG-representation.  Repeating this for all of $R\cup U$
gives $\log n-1$ centered {\singlevertical} $B_2$-VPG-graphs, to
which we can add the one defined by $M$.
\qed
\end{proof}

\subsection{Putting it all together}

We summarize all these results in our main result about splits into co-comparability
graphs:

\begin{theorem}
\label{thm:main2}
Let $G$ be a $B_2$-VPG-graph.
Then the vertices of $G$
can be partitioned into at most $\max\{1,8\log^3 n\}$ 
sets such that the subgraph
induced by each is co-comparability graph of poset dimension 3.
If $G$ is a 1-string $B_2$-VPG graph, then the subgraphs
are permutation graphs.
\end{theorem}
\begin{proof}
The claim is trivial for small $n$ since then $n\leq 8 \log^3 n$,
so assume $n\geq 4$.  Fix a $B_2$-VPG-representation $R$.
First split $R$ into two {\singlevertical} $B_2$-VPG-representations as in Lemma~\ref{lem:B2VPGsinglevertical}.
Split each of them into $\log n$ {\singlevertical} centered $B_2$-VPG-representations using
Lemma~\ref{lem:centered}, for a total of at most $2\log n$ sets of curves.
Split each of them into $4\log^2 n$ co-comparability graphs (or permutation graphs if
the representation was 1-string) using Lemma~\ref{lem:cocompar} or Corollary~\ref{cor:permutation}.
The result follows.
\qed
\end{proof}

We can do better for $B_1$-VPG-graphs.  The subgraphs obtained in the
result below are the same ones that were used implicitly in the
$4\log^2 n$-approximation algorithm given by Lahiri et al.\cite{cit:cocoa}.

\begin{theorem}
\label{thm:main3}
Let $G$ be a $B_1$-VPG-graph.
Then the vertices of $G$
can be partitioned into at most $\max\{1,4\log^2 n\}$ 
sets such that the subgraph
induced by each is a permutation graph.
\end{theorem}
\begin{proof}
The claim is trivial if $n=1$, so assume $n>1$.
Fix a $B_1$-VPG-representation $R$, and
split it into $\log n$ {\singlevertical} centered $B_1$-VPG-representations 
using Lemma~\ref{lem:centered}.  Split each of them into two
centered $B_1$-VPG-representations, one of those curves with the horizontal
segment above the centering line, and one with the rest.  Each of the
resulting $2\log n$ centered $B_1$-VPG-representations is now grounded
(possibly after a $180^\circ$ rotation).  We can split each of them
into $2\log n$ permutation graphs using
Corollary~\ref{cor:B1cornered}, for a total of $4\log^2 n$ permutation
graphs.
\end{proof}

\section{Applications}

We now show how Theorem~\ref{thm:main} and \ref{thm:main2}
can be used for improved approximation algorithms for $B_2$-VPG-graphs.   
The techniques
used here are virtually the same as the one by Keil and Stewart
\cite{KeilStewart2006} and require two things.
First, the problem considered needs to be solvable on the special graphs class
(such as outerstring graph or co-comparability graph or permutation
graph) that we use.  Second, the problem must be {\em hereditary}
the sense that a solution in a graph implies a solution
in an induced subgraphs, and solutions in induced subgraphs can be used to
obtain a decent solution in the original graph.

We demonstrate this in detail using
weighted independent set, which Keil et al.~showed to be polynomial-time 
solvable in outer-string graphs \cite{cit:keil}.  Recall that this
is the problem, given a graph with vertex-weights, of finding a 
subset $I$ of vertices that has no vertices between them such
that $w(I):=\sum_{v\in I} w(v)$ is maximized, where $w(v)$ denotes
the weight of vertex $v$.  The run-time to solve weighted independent
set in outerstring graphs is $O(N^3)$, where $N$ is the number of segments
in the given outer-string representation. 

\begin{theorem} 
\label{thm:main2}
There exists a $(2\log n)$-approximation algorithm for 
weighted independent set on {\singlevertical} graphs with
run-time $O(N^3)$, where $N$ is the total number of segments
used among all {\singlevertical} objects.
\end{theorem}
\begin{proof}
If $n=1$, then the unique vertex is the maximum weight independent set.
Else,
use Theorem~\ref{thm:main} to partition the vertices of the given graph $G$
into at most $2\log n$ sets, each of which induces an outer-string
graph.  This takes $O(N+n\log n)$ time, where
$N$ is the total number of segments of the representation of $G$.

Now
solve the weighted independent set problem in each subgraph $G_i$ by
applying the algorithm of Keil et al.  
If $G_i$ had an outer-string representation with $N_i$ segments in total,
then this takes time $O(\sum N_i^3)$ time.  Note that if a {\singlevertical}
object consisted of one vertical and $\ell$ horizontal segments, then we
can trace around it with a curve with $O(\ell)$ segments.  Hence all
curves together have $O(N)$ segments and the
total run-time is $O(N^3)$.

Let $I^*_i$ be the maximum-weight independent set in $G_i$, and
return as set $I$ the set in $I^*_1,\dots,I^*_k$ that has the
maximum weight.  To argue the approximation-factor, let
$I^*$ be the maximum-weight independent set of $G$,
and define $I_i$ to be all those elements of $I^*$ that belong to $R_i$,
for $i=1,\dots,k$.  Clearly $I_i$ is an independent set of $G_i$, and
so $w(I_i)\leq w(I_i^*)$.  But on the other hand $\max_i w(I_i) \geq w(I^*)/k$
since we split $I^*$ into $k$ sets.  Therefore $w(I)=\max_i w(I_i^*)
\geq w(I^*)/k$, and so the returned independent set is within a factor
of $k\leq 2\log n$ of the optimum.
\qed
\end{proof}

We note here that the best algorithm for independent set in general string
graphs achieves an approximation factor of $O(n^\varepsilon)$, under
the assumption that any two strings cross each other at most a constant
number of times \cite{FoxPach2011}.  This algorithm only works for unweighted
independent set; we are not aware of any approximation results for
weighted independent set in arbitrary string graphs.

The reader may wonder what types of graphs are single-vertical graphs.
It is not hard to show that all planar
graphs are single-vertical graphs (use a representation with touching
$T$'s \cite{FMR94}), and so are all graphs of boxicity 2 (i.e.,
intersection graphs of axis-aligned boxes) and intersection graphs
of disks in the plane.  Unfortunately, for these special graph
classes, the above theorem is no improvement over existing algorithms
for weighted independent set \cite{Baker94,Chan2003,CC09}.

Because $B_2$-VPG-graphs can be vertex-split into two {\singlevertical}
$B_2$-VPG-representations, and the total number of segments used is
$O(n)$, we also get:

\begin{corollary} 
\label{thm:main2}
There exists a $(4\log n)$-approximation algorithm for 
weighted independent set on $B_2$-VPG-graphs with
run-time $O(n^3)$.
\end{corollary}

Another hereditary problem is {\em colouring}: Find the minimum number
$k$ such that we can assign numbers in $\{1,\dots,k\}$ to vertices such
that no two adjacent vertices receive the same number.
Fox and Pach \cite{FoxPach2011} pointed out that if we have a $c$-approximation
algorithm for Independent Set, then we can use it to obtain an 
$O(c\log n)$-approximation algorithm for colouring.  Therefore our
result also immediately implies an $O(\log^2 n)$-approximation algorithm
for colouring in {\singlevertical} graphs and $B_2$-VPG-graphs.

Another hereditary problem is {\em weighted clique}: Find the maximum-weight subset
of vertices such that any two of them are adjacent.  (This is independent set
in the complement graph.)  We are not aware of any algorithms to solve weighted
clique in outerstring graphs (but it is also not known to be NP-hard).  For
this reason, we use the split into co-comparability graphs instead; weighted
clique can be solved in quadratic time
in co-comparability graphs (because weighted independent
set is linear-time solvable in comparability graphs \cite{Gol80}).  Weighted
clique is also linear-time solvable on permutation graphs \cite{Gol80}.
We therefore have:

\begin{theorem}
There exists an $(8\log^3 n)$-approximation algorithm for 
weighted clique on $B_2$-VPG-graphs with run-time $O(n^2)$.
The run-time becomes $O(n)$ if the graph is a 1-string $B_2$-VPG graph,
and the approximation factor becomes $4\log^2 n$ if the graph is
a $B_1$-VPG-graph.
\end{theorem}

In a similar manner, we can get poly-time $(8\log^3 n)$-approximation algorithms
for clique cover, maximum $k$-colourable subgraph, and maximum
$h$-coverable subgraph.  See \cite{KeilStewart2006} for the definition
of these problems, and the argument that they are hereditary.

\section{Conclusions}

We presented a technique for decomposing {\singlevertical} graphs into outer-string
subgraphs, $B_2$-VPG-graphs into co-comparability graphs, and 1-string $B_2$-VPG-graphs
into permutation graphs.  We then used these results to obtain approximation algorithms
for hereditary problems, such as weighted independent set.

We close with some open problems:
\begin{itemize}
\item Can we use a different method of splitting the representations to devise better
	approximation algorithms for $B_2$-VPG-graphs?  In particular, can we find
	an $O(1)$-approximation algorithm, or maybe even a PTAS, for independent set?
	Or is this problem APX-hard in $B_2$-VPG graphs?
\item Can we use a different method of combining the subgraphs to use such splits
	for problems that are not hereditary, but that are local in some sense?  For
	example, can we find a polylog-approximation algorithm for vertex cover or dominating
	set?
\item We can argue that a similar splitting technique can be used to split graphs with
	a $B_k$-VPG-representation for which all curves are monotone in both $x$-direction
	and $y$-direction.  But this is rather restrictive, and the number of subgraphs
	is rather large ($O(f(k)\log^k n)$ for some function $f(k)$).   Are there 
	poly-log approximation algorithms for, say, independent set in $B_k$-VPG-graphs
	for $k\geq 3$?
\end{itemize}

Last but not least, orthogonality was crucial for all our splits.  If curves are allowed
to have up to $k$ bends, but are not restricted to use horizontal or vertical lines,
are there any approximation algorithms better than the $O(n^\varepsilon)$-factor proved
by Fox and Pach \cite{FoxPach2011}?  Even for $k=0$ (i.e., intersection graphs of segments) this problem
appears wide open.

\bibliography{recursive}{}

\begin{thebibliography}{10}

\bibitem{Baker94}
B.~Baker.
\newblock Approximation algorithms for {NP}-complete problems on planar graphs.
\newblock {\em J. ACM}, 41(1):153--180, 1994.

\bibitem{CFM+16}
Jean Cardinal, Stefan Felsner, Tillmann Miltzow, Casey Tompkins, and Birgit
  Vogtenhuber.
\newblock Intersection graphs of rays and grounded segments.
\newblock Technical Report 1612.03638 [cs.DM], ArXiV, 2016.

\bibitem{CC09}
Parinya Chalermsook and Julia Chuzhoy.
\newblock Maximum independent set of rectangles.
\newblock In Claire Mathieu, editor, {\em Proceedings of the Twentieth Annual
  {ACM-SIAM} Symposium on Discrete Algorithms, {SODA} 2009, New York, NY, USA,
  January 4-6, 2009}, pages 892--901. {SIAM}, 2009.

\bibitem{Chan2003}
Timothy~M. Chan.
\newblock Polynomial-time approximation schemes for packing and piercing fat
  objects.
\newblock {\em J. Algorithms}, 46(2):178--189, 2003.

\bibitem{FMR94}
H.~de~Fraysseix, P.~Ossona de~Mendez, and P.~Rosenstiehl.
\newblock On triangle contact graphs.
\newblock {\em Combinatorics, Probability and Computing}, 3:233--246, 1994.

\bibitem{FoxPach2011}
Jacob Fox and J{\'{a}}nos Pach.
\newblock Computing the independence number of intersection graphs.
\newblock In Dana Randall, editor, {\em Proceedings of the Twenty-Second Annual
  {ACM-SIAM} Symposium on Discrete Algorithms, {SODA} 2011, San Francisco,
  California, USA, January 23-25, 2011}, pages 1161--1165. {SIAM}, 2011.

\bibitem{Gol80}
M.~C. Golumbic.
\newblock {\em Algorithmic graph theory and perfect graphs}.
\newblock Academic Press, New York, 1st edition, 1980.

\bibitem{GolumbicRU1983}
Martin~Charles Golumbic, Doron Rotem, and Jorge Urrutia.
\newblock Comparability graphs and intersection graphs.
\newblock {\em Discrete Mathematics}, 43(1):37--46, 1983.

\bibitem{cit:keil}
J.~Mark Keil, Joseph S.~B. Mitchell, Dinabandhu Pradhan, and Martin Vatshelle.
\newblock An algorithm for the maximum weight independent set problem on
  outerstring graphs.
\newblock {\em Comput. Geom.}, 60:19--25, 2017.

\bibitem{KeilStewart2006}
J.~Mark Keil and Lorna Stewart.
\newblock Approximating the minimum clique cover and other hard problems in
  subtree filament graphs.
\newblock {\em Discrete Applied Mathematics}, 154(14):1983--1995, 2006.

\bibitem{cit:cocoa}
Abhiruk Lahiri, Joydeep Mukherjee, and C.~R. Subramanian.
\newblock Maximum independent set on {$B_1$-VPG} graphs.
\newblock In Zaixin Lu, Donghyun Kim, Weili Wu, Wei Li, and Ding{-}Zhu Du,
  editors, {\em Combinatorial Optimization and Applications - 9th International
  Conference, {COCOA} 2015, Houston, TX, USA, December 18-20, 2015,
  Proceedings}, volume 9486 of {\em Lecture Notes in Computer Science}, pages
  633--646. Springer, 2015.

\end{thebibliography}
\bibliographystyle{plain}

\end{document}